%% file: main.tex
\documentclass[a4paper,UKenglish,cleveref, autoref, thm-restate]{lipics-v2021}

\pdfoutput=1 %
\hideLIPIcs  %

\title{Online sorting and online TSP:\texorpdfstring{\\}{} randomized, stochastic, and high-dimensional}

\titlerunning{Online sorting and online TSP: randomized, stochastic, and high-dimensional} 

\author{Mikkel Abrahamsen}{University of Copenhagen, Denmark}{miab@di.ku.dk}{https://orcid.org/0000-0003-2734-4690}{Supported by Starting Grant 1054-00032B from the Independent Research Fund Denmark under the Sapere Aude research career programme and part of Basic Algorithms Research Copenhagen (BARC), supported by the VILLUM Foundation grant 16582.}
\author{Ioana O. Bercea}{KTH Royal Institute of Technology, Stockholm, Sweden}{bercea@kth.se}{https://orcid.org/0000-0001-8430-2441}{}
\author{Lorenzo Beretta}{University of California, Santa Cruz, USA}{lorenzo2beretta@gmail.com}{https://orcid.org/0000-0002-4676-9777}{}
\author{Jonas Klausen}{University of Copenhagen, Denmark}{jokl@di.ku.dk}{https://orcid.org/0000-0002-7403-417X}{}
\author{László Kozma}{Institut für Informatik, Freie Universität Berlin, Germany}{laszlo.kozma@fu-berlin.de}{https://orcid.org/0000-0002-3253-2373}{Supported by DFG grant
KO 6140/1-2.}

\authorrunning{M. Abrahamsen, I.\,O. Bercea, L. Beretta,  J. Klausen, L. Kozma } 

\Copyright{Mikkel Abrahamsen, Ioana O. Bercea, Lorenzo Beretta, Jonas Klausen, László Kozma} 

\funding{\emph{Ioana O. Bercea, Lorenzo Beretta and Jonas Klausen:} Supported by grant 16582,  Basic Algorithms Research Copenhagen (BARC), from the VILLUM Foundation.}



\category{} %

\relatedversion{} %

\nolinenumbers %

\EventEditors{Timothy Chan, Johannes Fischer, John Iacono, and Grzegorz Herman}
\EventNoEds{4}
\EventLongTitle{32nd Annual European Symposium on Algorithms (ESA 2024)}
\EventShortTitle{ESA 2024}
\EventAcronym{ESA}
\EventYear{2024}
\EventDate{September 2--4, 2024}
\EventLocation{Royal Holloway, London, United Kingdom}
\EventLogo{}
\SeriesVolume{308}
\ArticleNo{56}

\usepackage{mathtools}
\usepackage{xfrac}
\newcommand{\DeclareAutoPairedDelimiter}[3]{%
  \expandafter\DeclarePairedDelimiter\csname Auto\string#1\endcsname{#2}{#3}%
  \begingroup\edef\x{\endgroup
    \noexpand\DeclareRobustCommand{\noexpand#1}{%
      \expandafter\noexpand\csname Auto\string#1\endcsname*}}%
  \x}

\usepackage{mdframed}  %

\input{defines}

\begin{document}

\maketitle

\begin{abstract}
In the \emph{online sorting problem}, $n$ items are revealed one by one and have to be placed (immediately and irrevocably) into empty cells of a size-$n$ array. The goal is to minimize the sum of absolute differences between items in consecutive cells. This natural problem was recently introduced by Aamand, Abrahamsen, Beretta, and Kleist (SODA 2023) as a tool in their study of online geometric packing problems. They showed that when the items are reals from the interval $[0,1]$ a competitive ratio of $O(\sqrt{n})$ is achievable, and no deterministic algorithm can improve this ratio asymptotically.

\smallskip
In this paper, we extend and generalize the study of online sorting in three directions:
\smallskip
\begin{itemize}
\item \emph{randomized}: we settle the open question of Aamand et al. by showing that the $O(\sqrt{n})$ competitive ratio for the online sorting of reals cannot be improved even with the use of randomness; 
\smallskip
\item \emph{stochastic}: we consider inputs consisting of $n$ samples drawn uniformly at random from an interval, and give an algorithm with an improved competitive ratio of $\widetilde{O}(n^{1/4})$. The result reveals connections between online sorting and the design of efficient hash tables;
\smallskip
\item \emph{high-dimensional}: we show that $\widetilde{O}(\sqrt{n})$-competitive online sorting is possible even for items from $\mathbb{R}^d$, for arbitrary fixed $d$, in an adversarial model. This can be viewed as an online variant of the classical TSP problem where tasks (cities to visit) are revealed one by one and the salesperson assigns each task (immediately and irrevocably) to its timeslot. Along the way, we also show a tight $O(\log{n})$-competitiveness result for \emph{uniform metrics}, i.e., where items are of different \emph{types} and the goal is to order them so as to minimize the \emph{number of switches} between consecutive items of different types. 
\end{itemize}

\end{abstract}

\newpage

\clearpage
\setcounter{page}{1}

\section{Introduction}\label{sec1}

The following natural problem called \emph{online sorting} was recently introduced by Aamand, Abrahamsen, Beretta, and Kleist~\cite{online_sort}: 
~Given a sequence $x_1,\dots,x_n$ of real values, %
assign them bijectively to array cells $A[1], \dots, A[n]$. 
Crucially, after receiving $x_j$, for $j=1,2,\dots,n$, we must immediately and irrevocably set 
$A[i] = x_j$, for some previously unused array index $i\in [n]$.  %
The goal is to minimize $\sum_{i=1}^{n-1}{\left| A[i+1] - A[i] \right|}$, the sum of absolute differences between items in consecutive cells. 

Aamand et al.~\cite{online_sort} study the problem as modelling certain geometric online packing problems. In particular, they use it to show lower bounds on the competitive ratio of such problems\footnote{The \emph{competitive ratio} of an online algorithm is the worst-case ratio of its cost to the optimum (offline) cost over inputs of a certain size $n$. The competitive ratio of a problem is the best competitive ratio achievable by an online algorithm for the problem. }. It is easy to see that the optimal (offline) solution of online sorting is to place the entries in sorted (increasing) order, and the question is how to approximate this solution in an online setting where the items are revealed one by one.

The problem evokes familiar scenarios where we must commit \emph{step-by-step} to an ordering of items, such as when scheduling meetings in a calendar, writing recipes in a notebook, planting trees in a garden, or writing data into memory cells. In such situations, some local coherence or sortedness is often desirable, but once the location of an item has been assigned, it is expensive to change it; for instance, it may be difficult to reschedule meetings or to migrate data items once memory locations are referenced from elsewhere. One must then carefully balance between placing similar items next to each other and leaving sufficiently large gaps amid uncertainty about future arrivals\footnote{The cost measure of online sorting is also natural as a measure of the \emph{unsortedness} of a sequence. For this, the \emph{number of inversions} (i.e., the number of pairs $a<b$ where $b$ appears before $a$) is perhaps more widely used. 
We argue, however, that with this latter cost, one cannot obtain a nontrivial competitive ratio. 
Consider an adversary that outputs $n/2$ copies of $0.5$, followed by either all $0$s or all $1$s until the end. It is easy to see that one of these choices results in at least $n^2/8$ inversions for any online algorithm, whereas the optimal (sorted) sequence has zero inversions.}.

Moreover, online sorting can be firmly placed among familiar and well-studied online problems; we briefly mention two. As they differ from online sorting in crucial aspects, a direct transfer of techniques appears difficult.
\begin{itemize}
\item \emph{list labeling} or \emph{order maintenance}~\cite{itai1981sparse, dietz1982maintaining, bender2002two, saks2018online, bender2022online}: in this problem, a sequence of values are to be assigned labels consistent with their ordering (in effect placing the values into an array). The main difference from online sorting is that the sequence must be fully sorted and the goal is to minimize \emph{recourse}, i.e., movement of previously placed items. 
\item \emph{matching on the line}~\cite{gupta2012online, raghvendra2018optimal, gupta2019stochastic, balkanski2023power}: here, a sequence of clients (e.g., drawn from $[n]$) are to be matched uniquely and irrevocably to servers (say, locations in $[n]$). The problem differs from online sorting mainly in its cost function; the goal here is to minimize the sum of distances between matched client-server pairs. 
\end{itemize}
One of the main results of Aamand et al.~\cite{online_sort}  is an algorithm for online sorting with competitive ratio $O(\sqrt{n})$. Aamand et al.\ require the input entries to come from the unit interval $[0,1]$ and to contain the endpoints $0$ and $1$. Conveniently, this makes the offline cost equal to $1$ and the competitive ratio equal to the online cost. We show that the same competitiveness result can be obtained even if these assumptions are relaxed. 

Aamand et al.~\cite{online_sort} also show that the $O(\sqrt{n})$ bound cannot be improved by any deterministic algorithm, and leave it as an open question whether it can be improved through randomization (assuming that the adversary is oblivious, i.e., that it does not know the ``coin flips'' of our algorithm).
Note that in several online problems such as \emph{paging} or \emph{$k$-server}, randomization can lead to asymptotic improvements in competitiveness (against an oblivious adversary);  %
e.g., see~\cite{Borodin}. As our first main result, we show that for online sorting, randomization (essentially) does not help.
\begin{restatable}{theorem}{rand1}\label{thm:rand1}
The (deterministic and randomized) competitive ratio of online sorting is $\Theta(\sqrt{n})$. 
The lower bound $\Omega(\sqrt{n})$ holds even when the input numbers are from $[0,1]$.
\end{restatable}

\mypara{Online TSP.} 
Ordering real values with the above cost (sum of differences between consecutive items) can be naturally viewed as a one-dimensional variant of the Traveling Salesperson Problem (TSP). 
Suppose that $n$ cities are revealed one by one (with repetitions allowed), and a salesperson must decide, for each occurrence $c$ of a city, on a timeslot for visiting $c$, i.e., the position of $c$ in the eventual tour. The cost is then the length of the fully constructed tour\footnote{A small technicality is whether the salesperson must return to the starting point or not. The effect of this in our cost regime, however, is negligible.}.
Formally, given a sequence $(x_1,\dots,x_n)$ of items $x_i \in {S}$, for a metric space $S$ with distance function $d(\cdot,\cdot)$, the goal is to assign the items bijectively to array cells $A[1],\dots, A[n]$ %
in an online fashion such as to minimize %
$\sum_{i=1}^{n-1}{d\left(A[i],A[i+1]\right)}$. We refer to this problem\footnote{This model is sometimes referred to as the \emph{online-list model}; it has been considered mostly in the context of scheduling problems~\cite{pruhs2004online, FiatW99}. To our knowledge, TSP has not been studied in this setting before. Note however, that a different online model, called the \emph{online-time model} has been used to study TSP~\cite{ausiello2001algorithms,megow2012power, bjelde2020tight}. In that model new cities can be revealed while the salesperson is already executing the tour; this may require changing the tour on the fly. Results in the two models are not comparable.} as \emph{online TSP in $S$}.

\emph{Online TSP in $\mathbb{R}$} is exactly online sorting. A natural $d$-dimensional generalization is \emph{online TSP in $\mathbb{R}^d$}, with the Euclidean distance $d(\cdot,\cdot)$ between items.\footnote{One may also view this task as a form of \emph{dimensionality reduction}: we seek to embed a $d$-dimensional data set in a one-dimensional space (the array), while preserving some distance information.} A difficulty arising in dimensions two and above is that the optimal cost is no longer constant; in $\mathbb{R}^d$, 
the (offline) optimum may reach 
$\Theta(n^{1-\frac{1}{d}})$
even if the input points come from a unit box. Computing the optimum exactly is NP-hard even if $d=2$~\cite{tsp_hard}. %
Our second main result is an online algorithm whose cost is $O(n^{1-\frac{1}{d+1}})$ and a competitiveness guarantee close to the one-dimensional case. %

\begin{restatable}{theorem}{thmHighDUB} \label{thm3}
There is a deterministic algorithm for online TSP in $\mathbb{R}^d$ with competitive ratio $\sqrt{d} \cdot 2^d \cdot O(\sqrt{n \log{n}})$. %
\end{restatable}

As this setting includes online sorting as a special case, the lower bound of $\Omega(\sqrt{n})$ applies. 
A key step in obtaining Theorem~\ref{thm3} is the study of the \emph{uniform metric} variant of the problem, i.e., the case of distance function $d$ with $d(x,y) = 1$ if and only if $x \neq y$. This captures the natural problem where items (or tasks) fall into a certain number of \emph{types}, and we wish to order them such as to minimize the number of \emph{switches}, i.e., consecutive pairs of items of different types. For this case we prove a tight (deterministic and randomized) competitive ratio, independent of the number of different types, which may be of independent interest. Our algorithm is a natural greedy strategy; we analyze it by modelling the evolution of contiguous runs of empty cells as a coin-removal game between the algorithm and adversary.

\begin{restatable}{theorem}{thmCompUniform}
\label{thm:comp_uniform}
The competitive ratio of online sorting of $n$ items under the uniform metric is $\Theta(\log{n})$.
The upper bound $O(\log{n})$ is achieved by a deterministic algorithm and the lower bound $\Omega(\log{n})$ also holds for randomized algorithms.%
\end{restatable}

Aamand et al.\ also consider the setting where the array is only partially filled (placing $n$ items into $m>n$ cells). In this case, the cost is understood as the sum of distances $d(x,y)$ over pairs of items $x,y$ with no other item placed between them. %
We thus extend our previous result to arrays of a larger size, obtaining a tight characterization. 

\begin{theorem}\label{thm:comp_uniform_large}
The competitive ratio (deterministic and randomized) of online sorting of $n$ items with an array of size $\ceil{\gamma n}$, with $\gamma>1$, under the uniform metric is $\Theta(1 + \log{\frac{\gamma}{\gamma-1}})$.
\end{theorem}

\mypara{Stochastic input.}
Given the (rather large) $\Omega(\sqrt{n})$ lower bound on the competitive ratio of online sorting, %
it is natural to ask whether we can overcome this barrier by relaxing the worst-case assumption on the input. Such a %
viewpoint has become influential recently in an attempt to obtain more refined and more realistic guarantees in online settings (e.g., see~\cite{bwac_book, gupta2019stochastic}). A natural model is to view each item as drawn independently from some distribution, e.g., uniformly at random from a fixed interval. Our next main result shows an improved competitive ratio for such stochastic inputs: %

\begin{restatable}{theorem}{stochasticone}\label{thm:stochastic1}
There is an algorithm for online sorting of $n$ items drawn independently and uniformly at random from $(0,1]$ that achieves competitive ratio $O((n\log n)^{1/4})$ with probability at least $1-2/n$.  
\end{restatable}

The algorithm illuminates a connection between online sorting and hash-based dictionaries~\cite[\S\,6.4]{Knuth3}.
In the latter, the task is to place a sequence of $n$ keys in an array of size $O(n)$, with the goal of minimizing \emph{search time}. Fast searches are achieved by hashing keys to locations in the array. Due to hash collisions, not all elements can be stored exactly at their hashed location, and various paradigms have been employed to ensure that they are stored ``nearby'' (for the search to be fast). We adapt two such paradigms to online sorting. The first is to hash elements into buckets and solve the problem separately in each bucket. Each bucket has a fixed capacity, and so an additional \emph{backyard} is used to store keys that do not fit in their hashed bucket~\cite{ arbitman2010backyard,DBLP:conf/stoc/BenderFKKL22}. In~\Cref{thm:stochastic1}, we use the values of the entries to assign them to buckets and employ a similar backyard design. We note some critical technical differences: in our design, all %
buckets must be full and we solve the problem recursively in each bucket; we also operate in a much tighter balls-into-bins regime, as we are hashing $n$ elements into exactly $n$ locations.

When the array size is allowed to be bigger than $n$, we employ yet another way for resolving hash collisions: the linear probing approach of Knuth~\cite{knuth63linprobe}. Here, %
the value $\lceil \upalpha n \rceil$ serves as the hash location of an entry $\upalpha \in (0,1)$. Intuitively, this is a good approximation for where the entry would appear in the sorted order. If we make sure that we place the entry close enough to this intended location, we can hope for a small overall cost. That is, we use the fact that linear probing places similar values close to each other. 
We get the following:

\begin{restatable}{theorem}{stochasticLarge}\label{stochasticLarge}
  For any \(\gamma > 1\), there is an algorithm for online sorting of \(n\) items drawn independently and uniformly from \((0, 1)\) into an array of size \(\ceil{\gamma n}\) that achieves expected competitive ratio  \(O\left(1 + \frac{1}{\gamma-1}\right)\).
\end{restatable}

\mypara{Paper structure.} 
In \S\,\ref{sec2} we present our results for the standard (one-dimensional) online sorting, in particular the lower bound for the randomized competitive ratio (Theorem~\ref{thm:rand1}). Results for online TSP in $\mathbb{R}^d$ and results for the uniform metric (Theorems~\ref{thm3}, \ref{thm:comp_uniform}, \ref{thm:comp_uniform_large}) are in \S\,\ref{sec3}. Our results for online sorting with stochastic input (Theorems~\ref{thm:stochastic1}, \ref{stochasticLarge}) are in \S\,\ref{sec4}. We conclude with a list of open questions in \S\,\ref{sec5}.  We defer some proofs and remarks to the Appendix.

\section{Competitiveness for online sorting}\label{sec2}

Given a sequence $X = (x_1, \dots, x_n) \in \mathbb{R}^n$ and a bijection $f \colon [n] \rightarrow [n]$ assigning ``array cells'' $A[i] = x_{f(i)}$, %
let 
$\DD_f(X) = \sum_{i=1}^{n-1}{|A[i+1]-A[i]|}$. 
Let $\OPT(X)$ denote the quantity $\min_f{\DD_f(X)}$, i.e., the \emph{offline optimum}. 

An \emph{online} algorithm is one that constructs the mapping $f$ incrementally. Upon receiving $x_j$, for $j=1, \dots, n$, the algorithm immediately and irrevocably assigns $f(i) = j$, for some previously unassigned $i \in [n]$. 
For an online algorithm $\AA$, we denote by $\AA(X)$ the cost $\DD_f(X)$ for the function $f$ constructed by algorithm $\AA$ on input sequence $X$.
We are interested in the competitive ratio $\CC_\AA(n)$ of an algorithm $\AA$, i.e., the supremum of $\AA(X)/\OPT(X)$ over all inputs $X$ of a certain size $n$, and in the best possible competitive ratio $\CC = \CC(n) = \inf_\AA{\CC_\AA(n)}$ obtainable by any online algorithm $\AA$ for a given problem.

\mypara{Upper bound.} 
Aamand et al.~\cite{online_sort} show that $\CC \in \Theta{(\sqrt{n})}$, under the additional restrictions that $x_i \in [0,1]$ for all $i \in [n]$, and that $\{0,1\} \subseteq \{x_1,\dots,x_n\}$. As our first result, we employ a careful doubling strategy to show the same upper bound without the two restrictions, for general sequences of $n$ reals (the lower bound clearly continues to hold).

\begin{restatable}[Proof in~\cref{appb}]{claim}{thmClaimOne}
\label{thm:allreals}
There is a deterministic online algorithm $\AA$ for online sorting of an arbitrary sequence of $n$ reals, with $\CC_\AA \in O(\sqrt{n})$.
\end{restatable}

\mypara{Lower bound.}
The competitive ratio $\CC_{\AA}$ 
of a \emph{randomized} algorithm $\AA$ is the supremum of $\mathbb{E}[\AA(X)]/\OPT(X)$ over all inputs $X$, where the expectation is over the random choices of $\AA$. 
Aamand et al.\ \cite{online_sort} leave open the question of whether a lower bound of $\Omega(\sqrt{n})$ on the competitive ratio also holds for randomized algorithms.  
We settle this question in the affirmative. It is important to emphasize that the random choices of $\AA$ are not known to the adversary, i.e., we assume the \emph{oblivious} model~\cite{Borodin}.

\begin{claim} \label{randomizedLowerBound}
$\CC_{\AA} \in \Omega(\sqrt{n})$ for every (possibly randomized) online algorithm $\AA$. 
\end{claim}

In the remainder of the section we prove Claim~\ref{randomizedLowerBound}, which implies Theorem~\ref{thm:rand1}. As usual, to lower bound the performance of a randomized algorithm (a distribution over deterministic algorithms), we lower bound instead the performance of a deterministic algorithm on an (adversarially chosen) distribution over input sequences. %

\mypara{Input distribution.}
Assume for simplicity that $\sqrt{n}$ is an integer. Repeat the following until a sequence of length \(n\) is obtained.
With probability \(\frac{1}{2 \sqrt{n}}\) set the remainder of the sequence to $0$s.
With probability \(\frac{1}{2 \sqrt{n}}\) set the remainder of the sequence to $1$s.
With probability \(1 - \frac{1}{\sqrt{n}}\) offer $\frac{k}{\sqrt{n}}$ for $k=0,\dots,\sqrt{n}-1$ as the %
next \(\sqrt{n}\) elements of the sequence, and flip a new coin.
We refer to the \(\sqrt{n}\) elements produced by the third option as an \emph{epoch}. 
Notice that %
$\OPT \leq 1$, %
thus it is enough to lower bound the expected cost of the algorithm. 
 
\mypara{Notation.}
Let \(T\) be a partially filled array, \(|T|\) be the number of stored elements and \(G(T)\) be the number of ``gaps'', i.e., maximally contiguous groups of empty cells. 

Let \(f(T)\) be the minimum cost of \emph{filling} \(T\) by an online algorithm, when the input sequence is chosen according to the distribution described above, where the cost of two neighboring elements \(T[i], T[i+1]\) are accounted for when \(T[i]\) or \(T[i+1]\) is filled, whichever happens the latest.

More formally, define the cost of a partially filled array \(T\) to be
\[c(T) = \sum_{\substack{i=0 \\ T[i], T[i+1] \text{ are nonempty}}}^{n-1} \size{T[i] - T[i+1]},\]
and let \(\AA(T)\) be the final array produced by an online algorithm \(\AA\) when the remaining \(n-\size{T}\) elements of the input sequence are generated as described above.
Then \(f(T) = c(\AA(T)) - c(T)\) is the difference between the final cost and the cost of the partially filled array \(T\).
Thus \(f(T) = 0\) if \(|T| = n\), and \(\mathbb E[f(\emptyset)]\) is the value we wish to bound, using \(\emptyset\) for the empty array.

For mappings \(T_1, T_2\) where \(T_2[i] = T_1[i]\) for all indices where \(T_1[i]\) is nonempty we say that \(T_2\) can be obtained from \(T_1\).
For such a pair of mappings let \(c(T_1, T_2) = c(T_2) - c(T_1)\) be the cost of transforming \(T_1\) into \(T_2\) according to the way of accounting given above.

Let \(\mathcal T(T)\) be the set of mappings that can be obtained by inserting an epoch into \(T\).
That is, \(\mathcal{T}(T)\) contains all mappings \(T'\) which can be obtained from \(T\) where the difference between \(T'\) and \(T\) corresponds to the elements of an epoch.
Note that \(|T'| = |T| + \sqrt{n}\) for \(T' \in \mathcal T(T)\).

Let \(\mathcal{L} = \{ T \colon G(T) \leq \frac{\sqrt{n}}{8}\}\) and \(\mathcal{H} = \{ T \colon G(T) > \frac{\sqrt{n}}{8}\}\) be the sets of all partially filled arrays with a low/high number of gaps.

\begin{lemma}\label{lem_ratio}
    Let \(T_1, T_2 \in \LL\) with \(T_2 \in \mathcal T(T_1)\).
    Then \(c(T_1, T_2) \geq \frac{3}{16}\).
\end{lemma}
\begin{proof}
    If an element is placed between two empty cells, the number of gaps will increase by one.
    If it is placed next to one or two occupied cells, the number of gaps stays constant or is reduced by one -- call such an insertion an \emph{attachment}.
    As \(G(T_2) \leq G(T_1) + \frac{\sqrt{n}}{8}\) at least \(\frac{7\sqrt{n}}{16}\) of the \(\sqrt{n}\) insertions of the epoch must be attachments.
    
    An attachment incurs cost at least \(\frac{1}{\sqrt{n}}\) unless the value(s) in the neighboring cell(s) and the newly inserted value are identical.
    As all values within an epoch are distinct, an attachment can only be without cost if the neighboring cell was occupied in \(T_1\) (that is, before the epoch). As
    \(G(T_1) \leq \frac{\sqrt{n}}{8}\) at most \(\frac{\sqrt{n}}{4}\) occupied cells border an empty cell. The remaining \(\sqrt{n}(\frac{7}{16} - \frac{1}{4}) =  \frac{3\sqrt{n}}{16}\) attachments will incur non-zero cost.
\end{proof}

As a shorthand, let \((T+0)\) and \((T+1)\) be the mappings obtained by filling all gaps in \(T\) with $0$s/$1$s, respectively.
\begin{lemma}\label{lem_both}  $c(T, (T+0)) + c(T, (T+1)) \geq G(T). $
\end{lemma}
\begin{proof}
    Each gap in \(T\) is bordered by at least one non-empty cell \(A[i]\). We have 
    \(|A[i] - 0| + |A[i] - 1| = 1\).
\end{proof}

For a mapping \(T\) the expected remaining cost can be bounded from below by
\begin{multline*}
    \mathbb E[f(T)] \geq \frac{1}{2\sqrt{n}} \cdot c(T, (T+0)) + \frac{1}{2\sqrt{n}} \cdot c(T, (T+1)) \\
    + \left(1 - \frac{1}{\sqrt{n}}\right) \cdot \min_{T' \in \mathcal T(T)} \{ c(T, T') + \mathbb E[f(T')] \}
\end{multline*}
and by \cref{lem_both} we thus have
\begin{align} \label{eq:Eft}
    \mathbb{E}[f(T)] \geq \frac{G(T)}{2\sqrt{n}} + \left(1 - \frac{1}{\sqrt{n}}\right) \cdot \min_{T' \in \mathcal T(T)} \{ c(T, T') + \mathbb E[f(T')] \} \, .
\end{align}
\begin{align*}
\textrm{Let~}
    \LL(i) = \min_{\mathclap{\substack{T \in \LL \\ |T|=n - i\cdot\sqrt{n}} }} {~\mathbb E[f(T)]},  \mathrm{\quad and \quad }
    \HH(i) = \min_{\mathclap{ \substack{T \in \HH \\ |T|=n - i\cdot\sqrt{n}}}} {~\mathbb E[f(T)]}
\end{align*}
be the minimum expected cost of filling \emph{any} array with \(i \cdot \sqrt{n}\) empty cells, and which contains a low/high number of gaps, for $i\in \{0,1,\dots,\sqrt{n}\}$, respectively $i\in\{0,1,\dots,\sqrt{n}-1\}$. (Note that $\HH(\sqrt{n})$ is undefined as an empty array cannot have a high number of gaps.)  %

Combining \cref{eq:Eft} with \cref{lem_ratio} we obtain
\begin{align*}
    \LL(i) &\geq \left(1 - \frac{1}{\sqrt{n}}\right) \cdot \min \{{3}/{16} + \LL(i-1),\, \HH(i-1)\}, \\
    \HH(i) &\geq {1}/{16} + \left(1 - \frac{1}{\sqrt{n}}\right) \cdot \min \left\{\LL(i-1),\, \HH(i-1)\right\},
\end{align*}
with \(\LL(0) = \HH(0) = 0\). Our next \lcnamecref{thmClInd}, proved by induction (see \Cref{appb}), leads to the lower bound.

\begin{restatable}{lemma}{thmClInd} \label{thmClInd}
$\LL(\sqrt{n}) \in \Omega(\sqrt{n})$.
\end{restatable}

As the empty mapping \(\emptyset\) is contained in \(\LL\) we have \(\mathbb{E}[f(\emptyset)] \geq \LL(\sqrt{n}) \in \Omega(\sqrt{n})\).
Yao's minmax principle~\cite[Prop.~2.6]{motwani1995} implies the lower bound on the expected cost of randomized algorithms for a worst-case input. 

\section{Competitiveness for online TSP}\label{sec3}

We now consider the generalization of online sorting that we call \emph{online TSP}. Given a sequence $X = (x_1, \dots, x_n) \in S^n$ for some metric space $S$, and a bijection $f \colon [n] \rightarrow [n]$, and $A[i] = x_{f(i)}$, %
let 
$\DD_f(X) = \sum_{i=1}^{n-1}{d(A[i+1], A[i])}$. Here, $d(\cdot,\cdot)$ is a metric over $S$. 
As before, $\OPT(X) = \min_f{\DD_f(X)}$, i.e., the \emph{offline optimum}, and $\AA(X)$ is the cost $\DD_f(X)$ for a function $f$ constructed by an online algorithm $\AA$ on input sequence $X$.
We define $\CC_\AA$ and $\CC$ as before. 

Our main interest is in the case $S = \mathbb{R}^d$, and particularly $d=2$, with $d(\cdot,\cdot)$ the Euclidean distance. 
As a tool in the study of the Euclidean case, we first look at a simpler, \emph{uniform metric} problem in \S\,\ref{sec31}, showing a tight $\Theta(\log{n})$ bound on the competitive ratio. Then, in \S\,\ref{sec32} we study the Euclidean $\mathbb{R}^2$ and $\mathbb{R}^d$ cases. As the uniform metric case is natural in itself, we revisit it in \S\,\ref{sec34} in the setting where the array size is larger than $n$. 

\subsection{Uniform metric}\label{sec31}

Let $S$ be an arbitrary discrete set and consider the distance function $d(x,y)=0$ if $x=y$ and $d(x,y)=1$ otherwise. Let $K = 
K(X)$ denote the number of distinct entries in the input sequence $X$, i.e., $K = |\{x_1, \dots, x_n\}|$. We give instance-specific bounds on the cost in terms of $K$ and $n$. The following claim is obvious.

\begin{claim}\label{lemoptunif}
$\OPT(X) = K-1$.
\end{claim}

Next, we give a bound on the online cost and show that it is asymptotically optimal.

\begin{claim}\label{lemgame}
There is an online algorithm $\AA$ with cost $\AA(X) \leq K \log_2{n}.$
\end{claim}

\begin{claim}\label{lemlb}
For every (possibly randomized) algorithm $\AA$ and all $K \geq 3$, there is an input distribution $X$ with $K = K(X)$ such that $\Ep{\AA(X)} \in \Omega(K \log{n})$.
\end{claim}

Note that the condition $K\geq3$ is essential; if $K=1$, then $\AA(X) = \OPT(X) = 0$ for every algorithm $\AA$, and if $K=2$, then there is an online algorithm $\AA$ that achieves $\AA(X) = \OPT(X) = 1$. (Place the two types of elements at the opposite ends of the array.)

Claims~\ref{lemoptunif}, \ref{lemgame}, \ref{lemlb} together yield the main result of this subsection. 

\thmCompUniform*

\begin{proof}[Proof of Claim~\ref{lemgame}]

We describe algorithm $\AA$ (see Appendix~\ref{appc} for an alternative), noting that it %
is not required to know $K$ in advance. 
Assume without loss of generality that $\{x_1,\dots,x_n\} = \{1,2,\dots,K\}$. 
For each $j$, with $1 \leq j \leq K$, maintain a \emph{cursor} $c_j \in [n]$ indicating the array cell where the next item $x_i$ is placed if it equals $j$.  More precisely, if $x_i = j$, then let $f({c_j}) = i$, and move the cursor to the right: $c_j = \min\{c_j+1,n\}$.

If $f(c_j)$ is already assigned (i.e., the cell $A[c_j]$ is already written), set $c_j$ to the mid-point of the largest empty contiguous interval. %
Similarly, when $j$ is encountered the first time, then initialize $c_j$ at the mid-point of the largest empty contiguous interval. Thus, initially, $c_{x_1} = \lfloor {n/2} \rfloor$, i.e., place the first element at the middle of the array.

Clearly, $\AA$ can always place $x_i$ \emph{somewhere}, so $\AA$ correctly terminates. 
It remains to prove the upper bound on the number of unequal neighbors at the end of the process.

We model the execution of $\AA$ as a \textbf{coin-game}. Consider a number of up to $K$ \emph{piles} of coins. The game starts with a single pile of $n$ coins. An adversary repeatedly performs one of two possible operations:
(1) remove one coin from an arbitrary pile, (2) split the \emph{largest} pile into two equal parts. Operation (2) is only allowed when the number of piles is less than $K$, and only for a pile of at least two coins. The game ends when all coins have been removed. 

It is easy to see that this game models the execution of $\AA$ in the sense that for any execution of $\AA$ on $X$ there is an execution of the coin-game in which the sizes of the piles correspond at each step to the lengths of the contiguous empty intervals in the array. Moreover, the number of consecutive unequal pairs at the end of algorithm $\AA$ (i.e., the cost $\AA(X)$) is at most the number of splits (i.e., operations (2)) of the coin-game execution. It is thus sufficient to upper bound the number of splits in \emph{any} execution of the coin game. 

Let $n_i$ denote the size of a pile before its split, for the $i$-th split operation. As $n_i$ is the size of the largest pile and piles can only get smaller, the sequence $n_i$ is non-increasing. 
Suppose the $i$-th split replaces a pile of size $t$ with two piles of size $t/2$. Then, after the split there are at most $K-2$ piles with sizes in $[t/2,t]$ and no pile greater than $t$. Thus, after at most $K-2$ further splits, we split a pile of size at most $t/2$. (Possible operations (1) can only strengthen this claim, as they make some piles smaller.)
It follows that $n_{i+K} \leq n_i/2$ for all $i$. As $n_1 \leq n$, and $n_i \geq 1$ for all $i$, the number of splits is at most $K \log_2{n}$.
\end{proof}

\begin{proof}[Proof of \cref{lemlb}]
  Let $\A$ be an algorithm filling the items into an array $A$ of size $n$.
  We present a distribution over inputs $X$ incurring cost $\Ep{X} \geq \Omega(K \log n)$.
  By Yao's minmax principle~\cite[Prop.\ 2.6]{motwani1995} we thus get that every randomized algorithm has a worst-case input of cost $\Omega(K \log n)$.

  For each free cell $A[i]$ we say that $A[i]$ is \emph{friendly} for the two (possibly identical) elements $y_k$ and $y_l$ that are placed closest to the left and right of $A[i]$ in $A$.
When asked to place an element $y$ into $A$, $\A$ will increase the cost of the partial solution \emph{unless} $y$ is placed in a cell that is friendly for~$y$.
When placing $y$ in a cell that is friendly for $y$, the number of friendly cells for $y$ will decrease by one. The number of friendly cells for other elements may or may not decrease, but no element will have more friendly cells than before the insertion of~$y$.

Let $z_1 \leq z_2 \leq \cdots \leq z_K$
be the number of friendly cells for each of the $K$ values, sorted nondecreasingly.
With $n'$ free cells in $A$, $\sum_{i=1}^K z_i \leq 2n'$, and hence $z_1 \leq 2n' / K$.
Consider the median $z_{\floor{K/2}}$. As $\sum_{i=\floor{K/2}}^K z_i \leq 2n'$, we have $z_1 \leq \dots \leq z_{\floor{K/2}} \leq \frac{2n'}{\frac{1}{2}K} = 4n'/K$.

We now describe our input distribution, proceeding in \emph{epochs}:
At the start of each epoch a value $y$ is chosen uniformly at random from $\set{1, \dots, K}$, and this element is presented all through the epoch.
If $K \in \set{3, 4}$, the epoch will consist of $2n'/K$ copies of $y$, where $n'$ is the number of unoccupied cells at the start of the epoch.
By the first observation above, $\Prp{z_y \leq 2n'/K} \geq 1/K \geq 1/4$, in which case the input forces $\A$ to increase the cost of the partial solution.

If $K > 4$, we instead let the epoch be of length $4n'/K$, and by the second observation above we obtain that $\Prp{z_y \leq 4n'/K} \geq 1/2$, again increasing the cost of the partial solution with constant probability.

To establish a lower bound for the expected cost produced by this input, it remains to lower bound the number of epochs processed.
For small $K$, the epoch leaves $n' \cdot \frac{K-2}{K}$ free cells for the coming epochs. For $K > 4$ the epoch leaves $n' \cdot \frac{K - 4}{K}$ free cells. Thus, in both cases, the number of epochs is at least
\[
\log_{\frac{K}{K-4}} (n) = \frac{\log_2(n)}{\log_2 \left( 1 + \frac{4}{K-4} \right)}
\geq \frac{K-4}{4} \log_2(n) \, .
\]
\noindent %
Hence, the input will force \(\A\) to produce a solution of expected cost $\Omega(K \log n)$.
\end{proof}

\subsection{Online TSP in \texorpdfstring{$\mathbb{R}^d$}{R\^{}d}}\label{sec32}
We now proceed to the case where $S = \mathbb{R}^d$ and $d(\cdot,\cdot)$ is the Euclidean distance.
We start with the first new case, $d=2$. 
For ease of presentation, we omit floors and ceilings in the analysis. We assume that the input points are from the unit box $[0,1]^2$ and that the optimum length is at least $1$. These assumptions can be relaxed by a similar doubling-approach as in the proof of Claim~\ref{thm:allreals}.

The following result is well known~\cite{Beardwood_Halton_Hammersley_1959, DBLP:journals/siamdm/Karloff89,Few_1955} (consider, e.g., a $\sqrt{n} \times \sqrt{n}$ uniform grid). %

\begin{claim}\label{geombound}
For all sequences $X$ of $n$ points in $[0,1]^2$, we have $\OPT(X) \in O(\sqrt{n})$. Moreover, there exists a sequence $X$ of $n$ poins in $[0,1]^2$ such that $\OPT(X) \in \Omega(\sqrt{n})$.
\end{claim}

As a warm-up before our competitiveness result, we give an upper bound on the cost of an online algorithm and we show the tightness of this bound. The arguments extend the one-dimensional ones in a straightforward way. The proofs of the following can be found in~\cref{appc}.
\begin{restatable}{claim}{thmClaimOnlineUB}   \label{claim_online_UB}
There is an online algorithm $\AA$ such that $\AA(X) \in O(n^{2/3})$ for all $X \in [0,1]^2$.
\end{restatable}

\begin{restatable}{claim}{thmClaimLB}   \label{claimLB}
For every deterministic $\AA$ there is an input $X$ such that $\AA(X) \in \Omega(n^{2/3})$.
\end{restatable}

Now we move to the study of the competitive ratio. Note that the lower bound of Claim~\ref{randomizedLowerBound} immediately applies to our setting. %
Combined with Claim~\ref{claim_online_UB}, it follows that $\CC \in O(n^{2/3}) \cap \Omega(n^{1/2})$. In the following we (almost) close this gap, proving the following.

\begin{restatable}{claim}{thmPlanarUB} 
\label{thm2}
The competitive ratio of online TSP in $\mathbb{R}^2$ is $O(\sqrt{n \log{n}})$.
\end{restatable}

Partition the box $[0,1]^2$ into $t \times t$ boxes of sizes $1/t \times 1/t$, for $t$ to be set later.
Let $K$ %
be the number of boxes that are \emph{touched}, i.e., that contain some input point $x_i$.
We first need to lower bound the optimum.

\begin{lemma}\label{lemopt}
$\OPT(X) \geq K/4t$. %
\end{lemma}
\begin{proof}
If $K/4t<1$, we are done, as $\OPT(X) \geq 1$ by assumption.
Assume therefore $K \geq 4t$. Choose an arbitrary representative point of $X$ from each touched box. Observe that the optimum cannot be shorter than the optimal tour of the representatives (by triangle inequality).
Among any five consecutive vertices of the representatives-tour, two must be from non-neighboring boxes, thus the total length of the four edges connecting these points is at least $1/t$. It follows that at least $K/4$ edges have this length, yielding the claim. 
\end{proof}

Now we need an algorithm that does well in terms of $K$. For any two neighboring entries in the array we consider only whether they come from the same box. This allows us to reduce our problem to the uniform metric case (Claim~\ref{lemgame}).

\begin{claim}\label{claimdist}
There is an online algorithm $\AA$ such that $\AA(X) \in O( K\log{n} + n/t)$.
\end{claim}
\begin{proof}
We treat points from the same box as having the \emph{same value}. We run the algorithm for the uniform metric given in Claim~\ref{lemgame}, incurring a total number of $O(K \log{n})$ differing pairs of neighbors. For these pairs we account for a maximum possible cost of $\sqrt{2}$. All other neighboring pairs have the same value (= come from the same box), incurring a cost of at most $O(1/t)$ each, for a total of $O(n/t)$. This completes the proof. 
\end{proof}

Together with Lemma~\ref{lemopt}, this yields Claim~\ref{thm2}. Indeed, set $t = \sqrt{\frac{n}{\log{n}}}$. 
If $K \leq t$, then using $\OPT \geq 1$, we have $\AA/\OPT \leq \frac{K\log{n}+n/t}{1} \in O(\sqrt{n\log{n}})$.
If $K>t$, then using Lemma~\ref{lemopt}, $\AA/\OPT \leq \frac{K\log{n}+n/t}{K/4t} = 4t\log{n} + 4n/K \in O(\sqrt{n\log{n}})$.

\mypara{Online TSP in $\mathbb{R}^d$.} 
We now extend the bound on the competitive ratio (Claim~\ref{thm2}) to the higher dimensional case, also considering the dependence on $d$, leading to the claimed result. %

\thmHighDUB*

Analogously to Claim~\ref{geombound}, we first state an absolute bound on the optimal TSP cost in $d$ dimensions, treating $d$ as a constant~\cite{Beardwood_Halton_Hammersley_1959, DBLP:journals/siamdm/Karloff89,Few_1955}.

\begin{claim}\label{geombound2}
For all sets $X$ of $n$ points in $[0,1]^d$, we have $\OPT(X) \in O(n^{1-1/d})$. Moreover, there exists a set $X$ of $n$ poins in $[0,1]^d$ such that $\OPT(X) \in \Omega(n^{1-1/d})$.
\end{claim}

A straightforward generalization of Claims~\ref{claim_online_UB} and \ref{claimLB} yields (for all constant $d$):

\begin{claim}
There is an online algorithm $\AA$ for online TSP in $\mathbb{R}^d$ such that $\AA(X) \in O(n^{1-\frac{1}{d+1}})$ for all $X$. For every deterministic algorithm $\AA$ there is an input $X$ such that $\AA(X) \in \Omega(n^{1 - \frac{1}{d+1}})$.
\end{claim}

\begin{proof}[Proof of Theorem~\ref{thm3}.]
We follow the proof of Claim~\ref{thm2}, with  minor changes. We assume the input to come from the $d$-dimensional unit box $[0,1]^d$. We partition this box into $t^d$ boxes of sizes $(1/t)^d$. When adapting Lemma~\ref{lemopt}, we have to consider $2^d$ (instead of $4$) consecutive tour edges, thus obtaining $\OPT(X) \geq \frac{K}{2^d t}$. When adapting Claim~\ref{claimdist}, our upper bound increases by a factor of $\sqrt{d}$, the distance between antipodal vertices of a unit $d$-cube, replacing the implicit $\sqrt{2}$ in the earlier bound. Thus, when bounding the competitive ratio, we incur an overall factor of $\sqrt{d} \cdot 2^d$, yielding the result. 
\end{proof}

\subsection{Uniform metric with a larger array}\label{sec34}

We study the problem of placing \(n\) elements \(\{1, \dots, K\}\) into an array of size \(\ceil{\gamma n}\) for a fixed \(\gamma > 1\).
The algorithm and distribution from \cref{lemgame,lemlb} can be reused in this setting, stopping either process after the insertion of the first \(n\) elements.

\begin{restatable}{claim}{uniformLargeLB}
For every algorithm $\AA$ and any number $K \geq 3$ of input values, there is an input distribution \(X\) such that $\Ep{\AA(X)} \in \Omega(K \cdot (1+\log(\gamma/(\gamma-1))))$.
\end{restatable}
\begin{restatable}{claim}{uniformLargeUB}
There is an online algorithm $\AA$ with cost $\AA(X) \in O(K \cdot (1 + \log(\gamma/(\gamma-1)))).$
\end{restatable}
The two claims together imply Theorem~\ref{thm:comp_uniform_large}. We defer their proofs to \cref{proofs:uniformLarge}. They rely on similar arguments as the proofs of the claims mentioned above, with a refined argument for counting the number of epochs/moves performed.

\input{stochastic}

\input{stochastic-large}

\section{Conclusion and open questions}\label{sec5}

In this paper, we studied the online sorting problem in some of its variants. Several questions and directions to further study remain, we mention those that we find most interesting. 

\mypara{Online sorting of reals.} For large arrays ($m = \ceil{\gamma n}$, for $\gamma>1$) significant gaps remain between the upper and lower bounds on the competitive ratio of online sorting (see~\cite{online_sort}).

In a different direction, consider online sorting and \emph{online list labeling}, mentioned in \S\,\ref{sec1}.  %
These can be seen as two 
extremes of an ``error vs. recourse'' trade-off: online sorting allows no recourse, while list labeling allows no error. %
Achieving a smooth trade-off between the two optimization problems by some hybrid approach is an intriguing possibility. %

\mypara{Stochastic and other models.} In the stochastic setting of online sorting, we could improve upon the worst-case bound, but the obtained ratio (Theorem~\ref{thm:stochastic1}) is likely not optimal; we are not aware of nontrivial lower bounds. For stochastic online sorting with large arrays  (Theorem~\ref{stochasticLarge}), a matching lower bound is likewise missing. We restricted our study to uniform distribution over an interval; extending the results to other distributions remains open. 

Other models that go beyond the worst-case assumption could yield further insight; we mention two possible models: (1) online sorting \emph{with advice} (e.g., from a machine learning model): suppose that each input item comes with a possibly unreliable estimate of its rank in the optimal (sorted) sequence; what is the best way to make use of this advice? (2) online sorting \emph{with partially sorted input}: for instance, suppose that the \emph{input cost} $\sum_{i=1}^{n-1}{|x_{i+1}-x_i|}$ is small. What is the best guarantee for the online sorting cost in this case?

\mypara{Online TSP in various metrics.}
For online TSP in $\mathbb{R}^d$, for fixed $d$, a small gap remains between the lower bound $\Omega(\sqrt{n})$ and upper bound $O(\sqrt{n \log{n}})$ on the competitive ratio (Theorem~\ref{thm3}). The dependence on $d$ (when $d \in \omega(1)$) is not known to be optimal.  Online TSP in $\mathbb{R}^d$ with large array, and/or with stochastic input are also interesting directions. %

Some non-Euclidean metrics pose natural questions, e.g., $L_0$ or $L_{\infty}$ in $\mathbb{R}^d$, or tree- or doubling metrics. More broadly, is the optimal competitive ratio $O(\sqrt{n})$ for arbitrary metrics?

\bibliography{ref.bib}

\newpage

\appendix

\section{Proofs from \S\,\ref{sec2}}\label{appb}

\thmClaimOne*
\begin{proof}

Our task is to place $n$ values $x_1, \dots, x_n \in \mathbb{R}$ bijectively into the array cells $A[1]$, $\dots$, $A[n]$. We revisit and extend the algorithm from~\cite{online_sort}. The main obstacle is that in~\cite{online_sort} it can be assumed that $\OPT = 1$. In our case, $\OPT = \max_i{\{x_i\}} - \min_i{\{x_i\}}$; as this quantity is not known until the end, we estimate it based on the available data. 

For ease of presentation assume that \(x_1 = 0\).
This is without loss of generality as we can simply transform $x_i$ to $x_i-x_1$ for all $i=1,\dots,n$ and the costs of both offline and online solutions are unaffected. 

Start by reading $x_1, x_2$ and placing them arbitrarily (this incurs a cost of at most $2\OPT$), and set $q = \lceil \log_2{|x_2 - x_1|} \rceil = \lceil \log_2{|x_2|} \rceil$ (as we assumed $x_1=0$). Throughout the algorithm we maintain $q$ so that the input values seen so far are in the interval $J = [-2^q,2^q]$. This clearly holds after reading $x_2$ . For simplicity, assume first that $x_3,\dots,x_n \in J$, and $q$ does not need to be updated.

Let $N_1 = \lfloor \sqrt{n} \rfloor$ and $N_2 = 2N_1$. Partition the interval $J$ into $N_1$ equal-size subintervals $J_1, \dots, J_{N_1}$, called ``boxes'', and partition the array $A$ into $N_2$ equal-size subarrays $A_1, \dots, A_{N_2}$, called ``blocks'' (as $n$ is generally not divisible by $N_2$, we allow one block to be smaller). 

Whenever a value $x_i \in J_k$ is received, if it is the first value from box $J_k$, then we assign $J_k$ to a yet unassigned (empty) block $A_t$; we place $x_i$ in the leftmost cell of $A_t$ and continue placing future values from $J_k$ into $A_t$ (left-to-right).
If the block $A_t$ of $J_k$ is full, we assign a new empty block to $J_k$. 

If we run out of empty blocks, we must have processed at least $n/2$ values; to see this, observe that of the $N_2 = 2N_1$ blocks, at most $N_1-1$ are non-full, the rest have been filled. 
From this point on, we treat the remaining $n' \leq n/2$ empty cells of $A$ as a virtual (contiguous) array $A'$, recursively placing the remaining $n'$ values into $A'$ by the same strategy (we restart $\AA$ on $A'$, without the intitialization of $q$, i.e., we keep the current values of $q$ and $J$). 

The total cost $t(n)$ of $\AA$ is at most $$ t(n) \leq t(n')+2^{q+1} \cdot \frac{n}{N_1} + 2^{q^*+1} \cdot 2N_2.$$

The first term is the remaining cost of placing $n'$ values into a contiguous array of size $n'$.

The second term is the cost due to at most $n$ values already placed within full blocks. (As these come from one of the $N_1$ boxes of $J$ and $|J| \leq 2^{q+1}$, the bound follows.)

The third term is the cost due to boundaries between blocks, and ``at the interface'' with the recursive call, i.e., at the margins of incomplete blocks. (This term thus captures the cost of ``pretending'' $A'$ to be contiguous.) 
Here $2^{q^*+1}$ is an upper bound on the size of $J$ throughout the entire execution of $\AA$ (including future points), i.e., $q^{*} = \lceil \log_2{(\max_i{(x_i)} - \min_i{(x_i)})} \rceil = \ceil{\log_2{\OPT}}$, and $q \leq q^*$.

We have $t(n) \leq t(n/2) + \OPT \cdot O(\sqrt{n})$, which resolves to $t(n) \in \OPT \cdot O(\sqrt{n})$. Notice that $\OPT$ is the \emph{global optimum} of the problem, not the optimum in the current recursive call.  

\smallskip

Suppose now that we receive a point $x_i \notin J$. We update $$q = \lceil \log_2{(\max_{k\leq i}{(x_k)} - \min_{k \leq i}{(x_k)})} \rceil$$ and $J=[-2^{q}, 2^q]$.
Note that the value of $q$ increased, and all values \(x_1,\dots,x_{i}\) seen so far are also contained in the current \(J\). We then split $J$ into new boxes $J'_1, \dots, J'_{N_1}$ (without changing $N_1$ or $N_2$). The blocks $A_1, \dots, A_{N_2}$ remain, and if some old box $J_i$ was assigned to a block $A_j$, then we assign the corresponding new box $J'_i$ to $A_j$. With these changes, we continue with the algorithm. 

Observe that at the margin of every non-filled block we may end up with neighbors not from the same box (i.e., one is from the old, one is from the new). Thus, their distance may be up to $|J| = 2^{q+1}$, so we incur an additional cost of at most $2^{q+1} \cdot N_1 \leq 2^{q+1} \cdot \sqrt{n}$.

We may incur such an extra cost several times during the execution, but always for a different integer $q \leq q^{\star}$.
As $2^{q^\star - 1} \leq \OPT \leq 2^{q^\star + 1}$,
the total extra cost is at most $\sum_{q\leq q^{\star}}{ 2^{q+1} \cdot \sqrt{n}} \in \OPT \cdot O(\sqrt{n})$, and the overall bound is unaffected.
\end{proof}

\thmClInd*
\begin{proof}
Proceed by induction on \(i\). For all \(i' \in \{0, \dots, i-1\}\) assume
\begin{align*}
\HH(i') &\geq \left(1 - \frac{1}{\sqrt{n}}\right)^{i'} \cdot \frac{i'}{32}, \\
\LL(i') &\geq \left(1 - \frac{1}{\sqrt{n}}\right)^{i'} \cdot \frac{i'-1}{32}.
\end{align*}
Then
\begin{align*}
    \HH(i) &\geq \frac{1}{16} + \left(1 - \frac{1}{\sqrt{n}}\right) \cdot \min \{\LL(i-1),\, \HH(i-1)\} \\
    &\geq \frac{1}{16} + \left(1 - \frac{1}{\sqrt{n}}\right) \cdot \min\left\{ \left(1 - \frac{1}{\sqrt{n}}\right)^{i-1} \cdot \frac{i-2}{32},\, \left( 1 - \frac{1}{\sqrt{n}}\right)^{i-1} \cdot \frac{i-1}{32} \right\} \\
    &\geq \left( 1 - \frac{1}{\sqrt{n}}\right)^{i} \cdot \frac{i}{32},
\end{align*}
and
\begin{align*}
    \LL(i) &\geq \left(1 - \frac{1}{\sqrt{n}}\right) \cdot \min \left\{\frac{3}{16} + \LL(i-1),\, \HH(i-1)\right\} \\
    &\geq \left(1 - \frac{1}{\sqrt{n}}\right) \cdot \min\left\{ \frac{3}{16} + \left(1 - \frac{1}{\sqrt{n}}\right)^{i-1} \cdot \frac{i-2}{32},\, \left( 1 - \frac{1}{\sqrt{n}}\right)^{i-1} \cdot \frac{i-1}{32} \right\} \\
    &= \min\left\{\left(1 - \frac{1}{\sqrt{n}}\right) \cdot \frac{3}{16} + \left(1 - \frac{1}{\sqrt{n}}\right)^{i} \cdot \frac{i-2}{32},\, \left( 1 - \frac{1}{\sqrt{n}}\right)^{i} \cdot \frac{i-1}{32} \right\} \\
    &\geq \left(1 - \frac{1}{\sqrt{n}}\right)^i \cdot \frac{i-1}{32}.
\end{align*}

Finally, note that \(( 1 - \frac{1}{\sqrt{n}})^{\sqrt{n}} \geq 0.3\) for \(n \geq 10\), proving the claim.
\end{proof}

\section{Proofs from \S\,\ref{sec3}}\label{appc}

\subparagraph{An alternative algorithm for Claim~\ref{lemgame}.} We now present a different algorithm, resembling the one from~\cite{online_sort} and the one from the proof of Claim~\ref{thm:allreals}. If $K$ is known, then we can proceed in a straightforward way: Partition the array into $2K$ blocks, assign individual values to blocks, and when we run out of free blocks, recurse on the remaining free space. We get a cost of $t(n) = O(K) + t(n/2)$, yielding $t(n) \in O(K \log{n})$. Here the $O(K)$ term accounts for both the boundaries between blocks and the interface with the recursive call. 

We can handle the case of unknown $K$ by a doubling strategy, similar to the proof of \Cref{thm:allreals}.
Namely, start with $K=2$, and run the above algorithm. When a new element arrives that increases the number of different values seen to $K+1$, then set $K$ to $2K$, and split each block in the current recursive level into two equal parts. A full block will be replaced by two full blocks, an empty block will be replaced by two empty blocks. In both cases, there is no additional cost. 

A partially filled block will be replaced either by (1) a full block and an empty block, or (2) a full block and a partially filled block, or (3) a partially filled block and an empty block. In either case, if the block was assigned to a value $t$, then, if there is a resulting partially filled block, then it will continue to be filled with values $t$. Thus there is no incurred cost that is not accounted for already elsewhere, and the bound of $O(K \log{n})$ holds.

\thmClaimOnlineUB*
\begin{proof}
The algorithm closely follows the online sorting algorithm of~\cite{online_sort} and the extension in Claim~\ref{thm:allreals}.

Divide the box $[0,1]^2$ into $n^{1/3} \times n^{1/3}$ \emph{boxes} of sizes $n^{-1/3} \times n^{-1/3}$. Divide the array of size $n$ into $2n^{2/3}$ equal size \emph{blocks}. When a point arrives in a given box $B$, assign $B$ to a previously unassigned block $J$, and place points from $B$ into $J$ (in arrival order, left-to-right).
If $J$ is full, assign a new block to $B$.
If no more new blocks are available, re-use the empty space by treating it as if it were contiguous, and recursively assign blocks for the remaining (at most) $n/2$ input points. 

The total cost can be bounded as $t(n) \leq O(n^{2/3}) + t(n/2) \in O(n^{2/3})$, where we accounted for the at most $O(n^{-1/3})$ distance between any two neighbors, and the at most $O(1)$ distance between the $O(n^{2/3})$ block-boundaries and boundaries at the interface between the first call and the recursive call. %
\end{proof}

\thmClaimLB*
\begin{proof}
We adapt the lower bound strategy for online sorting from~\cite{online_sort}, with the $n^{2/3}$ allowed input points being the grid points of the $n^{1/3} \times n^{1/3}$ uniform grid spanning $[0,1]^2$. We present the adversary strategy against deterministic algorithms; a randomized oblivious adversary can be obtained using techniques similar to those in \S\,\ref{sec2}.

In the first phase, as long as there is a grid point that does not appear in the array with an empty neighbor, output that point as the next input item. If all grid points appear with an empty neighbor, then move to phase two, and output $(0,0)$ points until the end. In the first phase each input item incurs cost at least $\Omega{(n^{-1/3})}$, thus, if we remain in the first phase until the end, then the total cost is $\Omega(n \cdot n^{-1/3}) = \Omega{(n^{2/3})}$. If we move to the second phase, then the cost is at least $$\sum_{1 \leq i,j \leq n^{1/3}}{\sqrt{i\cdot n^{-1/3} + j \cdot n^{-1/3}}},$$ likewise yielding $\Omega{(n^{2/3})}$.
\end{proof}

\subsection{Proofs from \S\,\ref{sec34}}
\label{proofs:uniformLarge}
\uniformLargeLB*
\begin{proof} 
    We start by introducing one copy of each element \(\{1, \dots, K\}\), immediately causing a cost of \(K-1\).
    Then, we proceed as in the proof of \cref{lemlb}, choosing a type uniformly at random.
    Assume \(K \geq 5\).
    With \(m'\) free cells left in the array we repeat the chosen element \(4m'/K\) times.
    By the arguments presented in the proof of \cref{lemlb}, this increases the cost of the partial solution with probability at least \(1/2\).
    Next, a new element is chosen, and the process is repeated with \(m' ( 1 - 4/K)\) free cells left in the array.
    
    The question now is how many times we can repeat this procedure before having inserted \(n\) elements.
    Watching the process in reverse, it is seen that with \(m'\) cells currently free, we have caused an increase in cost (with probability 1/2) during the previous \(K m'/(K-4)\) insertions.
    The process is stopped when \(m' = (\gamma - 1) n\), and started at \(m' = \gamma n - K\).
    Hence we can repeat the process
    \begin{align*}
        \log_{\frac{K}{K-4}} \left( \frac{\gamma n - K}{(\gamma - 1) n} \right) &= \frac{ \log_2 \left( \frac{\gamma}{(\gamma - 1)}  - \frac{K}{(\gamma-1)n} \right) } { \log_2 \left( \frac{K}{K-4} \right) } \\
        &\geq \frac{K-4}{4} \cdot \log_2 \left( \frac{\gamma}{(\gamma - 1)}  - \frac{K}{(\gamma-1)n} \right) \\
        &\geq \Omega\left(K \cdot \log\left(\frac{\gamma}{\gamma-1}\right)\right)
    \end{align*}
    times.
    Thus the expected cost of \(\AA(X)\) will be \(K-1 + 1/2 \cdot \Omega(K \cdot \log(\gamma/(\gamma-1)))\).

    The case \(K \in \set{3, 4}\) is lifted from the proof of \cref{lemlb} in the same way: Letting each round be of length \(m'/K\), rounds will cause an increase in cost with probability at least \(1/K \geq 1/4\), and we can execute \(\Omega(K \cdot \log(\gamma/(\gamma-1)))\) rounds.
\end{proof}

\uniformLargeUB*
\begin{proof}
    We reuse the algorithm and coin-game from the proof of \cref{lemgame}, with the modification that the game starts with \(\gamma n\) coins in a single pile and ends when \((\gamma-1)n\) coins are left.
    Generalizing the observation from the proof of \cref{lemgame}, we observe that if we at some point have \(c\) piles each of size at most \(t\), then all piles will be of size at most \(t/2\) after another \(c\) splits.

    Applying this observation to the very start of the game, we initially have a single pile of \(\gamma n\) coins, then two piles of at most \(\gamma n/2\) coins each after the first split, then four piles of at most \(\gamma n/4\) coins after another two splits, and so forth.
    Generally, \(n_i \leq \gamma n / 2^{\floor{\log_2 i + 1}}\).
    
    From the above observation, \(n_{2K} \leq \gamma n/K\),
    and from the proof of \cref{lemgame} we have \(n_{i+K} \leq n_i / 2\).
    The early termination of the game means that \(n_i \geq (\gamma-1)n/K\) for all \(i\), thus
    \begin{align*}
        i &\leq 2K + K \log_2 \left( \frac{\gamma n}{K} \cdot \frac{K}{ (\gamma -1) n} \right) 
        = K \cdot \left(2 + \log_2 \left(\frac{\gamma}{\gamma - 1} \right) \right) \, . \qedhere
    \end{align*}
\end{proof}

\section{Proofs from \S\,\ref{sec4}}\label{appd}

\chernofffailure*
\begin{proof} Let $\size{A_i}$ denote the number of elements that hash (according to $h$) into a fixed bucket $A_i$, for some $i\in \set{1,\ldots, M}$. Then $\Ep{\size{A_i}} = n^{1-\alpha}$, since we have $n$ elements hashing into $M = n^{\alpha}$ buckets. Recall that the capacity of each bucket is $C = N/M$, where  $N=n-n^{\beta}$, and so $C = n^{1-\alpha} - n^{\beta-\alpha}$. We now apply a Chernoff bound for the lower tail of $\size{A_i}$:
	\begin{align*}
	\Prp{\size{A_i} < C} = \Prp{\size{A_i} < \parentheses{1- 1/n^{1-\beta}}   \cdot n^{1-\alpha}}  \leq& \exp\parentheses{ -1/2\cdot n^{1-\alpha} \cdot 1/n^{2-2\beta}} \\
    = & \exp\parentheses{-1/2 \cdot n^{2\beta-\alpha-1}} \;.
	\end{align*}

\noindent The claim follows by setting $\alpha$ and $\beta$ such that the above probability is at most $1/n^{c+\alpha}$ and then doing a union bound over all buckets.
\end{proof}

\optcostuniform*

\begin{proof} We first optimize for the cost of $\Sunif_1(A,n)$ from~\Cref{claim:costuniformone}. Note that the cost is (asymptotically) minimized when $n^{\beta} = \Theta(C)$, that is, when $n^{\beta} = \Theta(n^{1-\alpha})$. We thus set $\beta = 1-\alpha$. We now verify the conditions from~\Cref{claim:ballsintobins} and set 
	$$\alpha = \frac{1}{3} - \frac{\ln\ln n + \ln(2(c+1))}{3\ln n}  \;.$$
Then $n^{\frac{1-\alpha}{2}} = n^{\frac{1}{3}} \cdot n^{\frac{\ln\ln n + \ln(2(c+1))}{6\ln n}} = n^{\frac{1}{3}} \cdot \exp\parentheses{\frac{\ln\ln n + \ln(2(c+1))}{6}}$ and the claim follows.
\end{proof}

\thmBBk*
\begin{proof}
We prove the claim by induction on $k$. We assume that the statement is true for $k$ and we then prove it for $k+1$.  From the proof of~\Cref{claim:ballsintobins}, we have that the buckets in $\Sunif_{k+1}$ are not full with probability at most $1/n^{c}$ as long as 
\begin{equation}\label{eq:one}
	\beta \geq \frac{1}{2} \cdot \parentheses{1+ \alpha + \frac{\ln\ln n + \ln(2(c+1))}{\ln n} }\,.  
\end{equation}

We now would like to have the same upper bound for the probability that $\Sunif_k$ fails in some bucket. For that, we invoke the inductive hypothesis with failure probability at most $1/{n^{c+1}}$, and do a union bound over the $n^{\alpha}$ buckets. We implicitly assume here that $\alpha<1$, so that we have $2n^{\alpha} \leq n$. We upper bound $C$ by $n^{1-\alpha}$ and get that the cost for all the buckets is
$$	O\parentheses{ n^{(1-\alpha)/f_k} \cdot\exp\parentheses{1/4 \cdot \ln\ln n + B_k  \cdot \ln(2(c+2)) }  } \;.  $$

The cost across the buckets  and the backyard remains $O(1)$, and the cost of the backyard is $O(n^{\beta/2})$ as before. We now equate 

	$$\frac{1-\alpha}{f_k} + \frac{1/4 \cdot\ln\ln n + B_k\ln(2(c+2))}{ \ln n}= \frac{\beta}{2}$$
and then set
 
 $$  \alpha = \frac{4-f_k}{4+f_k} + \frac{f_k}{f_k+4}\cdot\frac{ 4B_k\cdot \ln(2c+4) - \ln(2c+2)  }{\ln n}$$
 to satisfy Eq.(\ref{eq:one}). We would then get that

$$ \frac{1-\alpha}{f_k} =  \frac{2}{4+f_k} - \frac{1}{4+f_k} \cdot \frac{ 4B_k\cdot \ln(2c+4) - \ln(2c+2)  }{\ln n}$$
 and the exponent in the overall cost would be 
  $$\frac{2}{4+f_k} +  \frac{1/4 \cdot \ln\ln n + B_k\cdot \ln(2c+4) - \frac{4B_k}{4+f_k}\cdot \ln(2c+4) + \frac{1}{4+f_k} \ln(2c+2)  }{\ln n} \;.$$
 
Note that  $2/(4+f_k) = 1/f_{k+1}$ by definition, so the only thing left to verify is that:

$$ B_k\cdot \ln(2c+4) - \frac{4B_k}{4+f_k}\cdot \ln(2c+4) + \frac{1}{4+f_k} \ln(2c+2)  \leq B_{k+1} \cdot \ln(2c+2) \;.$$ 

We upper bound $\ln(2c+4) \leq 2 \ln(2c+2)$ and get that it is enough to show that

$$ 2B_k \cdot \parentheses{1-\frac{4}{4+f_k}} + \frac{1}{4+f_k} \leq B_{k+1} \;.$$

Finally, we use the fact that $1/8 \leq 1/(4+f_k) \leq 1/4$ to show that this inequality holds when we set $B_k = k/4$. 
\end{proof}

\stochasticone*
\begin{proof}
 To get an overall bound, we invoke~\Cref{lem:ballsintobinsk} for $c=1$ and a value $k$ that leads to a small cost.  We set $x = 1/2^{k-1}$ and rewrite $k = \log_2 (2/x)$ in the main term in the cost as follows:
$$ n^{1/f_k} \cdot 4^{k/4} =n^{1/(4- x)}\cdot  \sqrt{2/x}=  \exp(\ln n/(4-x) + 1/2 \cdot \ln(2/x)) = \exp(\ln(2)/2 + f(x)) \;,$$
where we defined $f(x)  =  \ln n/(4-x) -1/2 \cdot \ln(x)$. This function is minimized at $x_0 = 4+\ln n - \sqrt{\ln^2 n+8\ln n}$, so we invoke $\Sunif_{k_0}(A,n)$ for $k_0=\log_2(2/x_0)$, yielding: %
\begin{align*}
	f(x_0) &= \frac{\ln n}{\sqrt{\ln^2 n+8\ln n} - \ln n} - \frac{1}{2} \cdot \ln (4+\ln n - \sqrt{\ln^2 n+8\ln n})%
 \leq  \frac{\ln n}{4}+\frac{1}{2} \;.
\end{align*}

\noindent Plugging $f(x_0)$ and $k_0$ back, we get an overall cost of $O((n \ln n)^{1/4})$. As the optimal (offline) cost is clearly $\Omega(1)$ with high probability, the claim follows. %
\end{proof}
 
\input{stochastic-large-appendix}

\end{document}

%% file: defines.tex
\newcommand{\calW}{\mathcal W}

\newcommand{\ceil}[1]{\left\lceil{#1}\right\rceil}
\newcommand{\floor}[1]{\left\lfloor{#1}\right\rfloor}

\newcommand{\eps}{\varepsilon}

\newcommand{\Prp}[1]{\Pr\!\left[{#1} \right]}

\newcommand{\Ep}[1]{\E\!\left[{#1} \right]}

\newcommand{\indicator}[1]{\left[{#1}\right]}

\newcommand{\set}[1]{\left \{ #1 \right \}}

\crefname{lemma}{Lemma}{Lemmas}
\Crefname{lemma}{Lemma}{Lemmas}
\crefname{theorem}{Theorem}{Theorems}
\Crefname{theorem}{Theorem}{Theorems}
\Crefname{corollary}{Corollary}{Corollaries}
\crefname{corollary}{Corollary}{Corollaries}
\crefname{observation}{Observation}{Observations}
\Crefname{observation}{Observation}{Observations}
\crefname{definition}{Definition}{Definitions}
\Crefname{definition}{Definition}{Definitions}
\crefname{section}{Section}{Sections}
\Crefname{section}{Section}{Sections}
\crefname{figure}{Figure}{Figures}
\Crefname{figure}{Figure}{Figures}
\crefname{appendix}{Appendix}{Appendices}
\Crefname{appendix}{Appendix}{Appendices}
\crefname{claim}{Claim}{Claims}
\Crefname{claim}{Claim}{Claims}

\newcommand{\size}[1]{\ensuremath{\left|#1\right|}}

\newcommand{\parentheses}[1]{\left(#1\right)}

\newcommand\drop[1]{}
\DeclareAutoPairedDelimiter\ceil{\lceil}{\rceil}
\DeclareAutoPairedDelimiter\floor{\lfloor}{\rfloor}

\renewcommand{\epsilon}{\varepsilon}
\newcommand{\E}{\mathbb{E}}

\usepackage{thmtools}
\usepackage{thm-restate}

\usepackage{upgreek}
\usepackage{tikz}
\usepackage{algpseudocode}
\usepackage{bbm}
\usepackage{microtype}

\newcommand{\LL}{\mathcal{L}}
\newcommand{\HH}{\mathcal{H}}

\newcommand{\DD}{\mathtt{D}}
\newcommand{\A}{\mathcal{A}}

\newcommand{\OPT}{\mathtt{OPT}}

\newcommand{\CC}{\mathtt{C}}
\renewcommand{\AA}{\mathcal{A}}

\newcommand{\mypara}[1]{\smallskip\noindent{\textbf{\sffamily #1}}}

%% file: stochastic.tex
\section{Competitiveness for stochastic online sorting}
\label{sec4}

In this section, we prove~\cref{thm:stochastic1} and~\cref{stochasticLarge} (\S\,\ref{sec41}).
\stochasticone*

We start by describing the general design of the algorithm from~\cref{thm:stochastic1} and showing some fundamental properties. We then describe a recursive algorithm and give the parameterizations that achieve the desired competitive cost.

\mypara{The general design.} We let $\alpha$ and $\beta$ denote two parameters in $(0,1)$ which we will set later.
We decompose the array $A$ of $n$ cells as such: the first $N = n-n^{\beta}$ cells are divided into $M = n^{\alpha}$ consecutive sub-arrays  $A_1,\ldots, A_{n^{\alpha}}$, and the remaining $n^{\beta}$ cells (at the end of the array) form one single subarray denoted by $B$. We refer to each of the sub-arrays $A_i$  as a \emph{bucket} and to $B$ as the \emph{backyard}. We note that each bucket $A_i$ can hold $C= N/M$ elements, which we refer to as the \emph{capacity} of the bucket.

We use the values of the elements to hash them into the array. Namely, for an element $x \in (0,1)$, we define $h \colon (0,1] \rightarrow \{1,\ldots, M\}$  by setting $h(x) = \ceil{x \cdot M}$. In other words, the elements in the interval $(0,1/M]$ will all hash to bucket $1$, elements in the interval $(1/M, 2/M]$ will hash to bucket $2$, etc. Since the elements are chosen independently and uniformly at random from $(0,1]$, we get that $h$ assigns the elements independently and uniformly at random into the $M$ buckets.\footnote{We assume that there is no element of value $0$. The probability of this happening is $0$.}

\newcommand{\Sdet}{\textsf{SortDet}}
\newcommand{\Sunif}{\textsf{SortUnif}}

\mypara{$\Sunif_1(A,n)$: the first algorithm.} Let $\Sdet(A,n)$ denote the deterministic algorithm from~\cite{online_sort}.
We now define the algorithm $\Sunif_1(A,n)$ as such: upon receiving $x$, it checks if the bucket $A_{h(x)}$ has any empty cells. If so, it forwards $x$ to $\Sdet(A_{h(x)},C)$. Otherwise, it forwards $x$ to $\Sdet(B, n^{\beta})$ (and we say that the corresponding bucket is \emph{full}). 

We first prove that $\Sunif_1(A,n)$ successfully places all items with high probability. Note that this is not always guaranteed: it could happen that 
both the bucket and the backyard are full (before we have managed to place all the $n$ elements). If this happens, then among all $n$ elements, strictly less than $C$ hash into some bucket. We call this event a \emph{failure} and show the following by employing a Chernoff bound (\cref{appd}):

\begin{restatable}{claim}{chernofffailure}\label{claim:ballsintobins} Given any $c>0$, if $\beta \geq \frac{1}{2} \cdot \parentheses{1+ \alpha + \frac{\ln\ln n + \ln(2(c+1))}{\ln n} }$, then $\Sunif_1(A,n)$ fails with probability at most $1/n^c$. 
\end{restatable}

\noindent We now bound the cost of $\Sunif_1(A,n)$:

\begin{claim}\label{claim:costuniformone} If $\Sunif(A,n)$ does not fail, then its cost is at most $O(\sqrt{C} + n^{\beta/2})$.
\end{claim}
\begin{proof}
	Since each bucket receives at least $C$ elements, the cost of placing elements inside bucket $A_1$ is given by the cost of $\Sdet(A_{1}, C)$ on elements from $(0,1/M)$. We bound this by $O(\sqrt{C} \cdot 1/M]$. For the remaining buckets, the elements are drawn from $(i/M, (i+1)/M)]$. This is equivalent to sorting elements from $(0,1/M]$, and so their cost is the same as that of $A_1$. Therefore, in total, the cost from each individual bucket is $O(\sqrt{C})$. In addition, we also have the cost from crossing from one bucket to the next. This is at most $2/M$, since the maximum difference between elements from consecutive buckets is at most $2/M$. Since there are $M$ buckets, this amounts to a cost of at most $2$. The cost of crossing from bucket $A_{M}$ to the backyard is $1$. Finally, the cost in the backyard is $O(n^{\beta/2})$, since it employs $\Sdet(B,n^{\beta})$ on elements from $(0,1)$.
\end{proof}

\noindent
We instantiate $\alpha$ and $\beta$ such that we get the following (proof in \cref{appd}):
\begin{restatable}{lemma}{optcostuniform}\label{lem:uniformone} Given any $c>0$, $\Sunif_1(A,n)$ has cost at most $O(n^{1/3} \cdot \ln^{1/6} n \cdot (2(c+1))^{1/6}  )$ with probability at least $1-1/n^c$.
\end{restatable}

\mypara{$\Sunif_k$: recursing on the buckets.} We take advantage of the fact that within each bucket, the elements are chosen uniformly at random. That is, we can apply the same strategy recursively inside the bucket. We get a series of algorithms $\Sunif_k$ for $k\geq 2$. In $\Sunif_k$, we let $\alpha$ and $\beta$ be defined as in $\Sunif_1$. When we see an element $x$ with $h(x)=i$, if its bucket $A_i$ is not full, we place $(xM-i+1)$ using $\Sunif_{k-1}(A_{h(x)}, C)$. Note that we have already conditioned on the fact that $ \ceil{x \cdot M} = i$. In this case,    $(xM-i+1)$ becomes uniformly distributed in $(0,1)$. If $A_i$ is full, we place $x$ in the backyard according to $\Sdet(B,n^{\beta})$. 

Note that there are now several causes for failure: either some bucket $A_i$ is not full, or one of the algorithms inside the bucket fails. We bound the probability of either of these happening as follows, proving the claim by induction over $k$ (\Cref{appd}):

\begin{restatable}{lemma}{thmBBk} \label{lem:ballsintobinsk} Given any $c>0$,  $\Sunif_k(A,n)$ has cost at most $$O\parentheses{ n^{1/f_k} \cdot \ln^{1/4} n\cdot (2(c+1))^{B_k}   } $$ with probability at least $1 - 2/n^{c}$, where $f_k = 4- 1/2^{k-1}$ and $B_k = k/4$.
\end{restatable}

Setting $c=1$ and $k=\log_2\parentheses{2/(4+\ln n-\sqrt{\ln^2 n+8\ln n})}$ yields the bounds claimed in \Cref{thm:stochastic1}; the details are given in \Cref{appd}.

%% file: stochastic-large.tex
\subsection{Stochastic online sorting in larger arrays}\label{sec41}
\stochasticLarge*

\mypara{The algorithm.}
Define \(\alpha = (\gamma-1)/10\) and \(\beta = \gamma - \alpha\) such that array \(A\) has size \((\beta + \alpha) n\).
We will refer to the final \(\alpha n\) cells of \(A\) as the \emph{buffer}.
The first \(\beta n\) cells we consider to correspond to the interval \([0, 1)\).
More specifically, cell \(A[i]\) represents the interval \(\left[\frac{i}{\beta n},\, \frac{i+1}{\beta n}\right)\).

Let \(h \colon [0, 1) \to \set{0, \beta n - 1}\) be the function which maps values \(x \in [0, 1)\) to the index of the corresponding cell: \(h(x) = \floor{\frac{x}{\beta n}}\).
Our algorithm \(\AA\) inserts each value \(x\) into \(A[h(x)]\) if the cell is available, and otherwise into the \emph{first} available cell following \(h(x)\), possibly wrapping around from \(A[\gamma n - 1]\) to \(A[0]\) (although, as we argue in \cref{sec:stochastic-large-appendix}, this is unlikely).
In other words, \(x\) is inserted into the cell \(A[h(x) + i \mod \gamma n]\) where \(i\) is the smallest non-negative integer such that the specified cell is available.

\mypara{Counting steps.}
As the values are drawn independently and uniformly from \([0, 1)\), the indices \(h(x)\) are also independently and uniformly distributed in \(\set{0, \dots, \beta n - 1}\).
  The algorithm thus mirrors the scheme of linear probing commonly used for implementing hash tables.
  In linear probing each key \(x\) that is to be added to the table (and which is generally not assumed to come from a known distribution) is hashed to an index of the array, and \(x\) is inserted into the first following free cell.
    
  When implementing a hash table, one is interested in analyzing the \emph{speed} of insertions, which for linear probing corresponds to the number of cells that are probed before a free cell is found.
  Although speed is not a concern for our online algorithm, we will nonetheless show that the number of steps performed is an important measure for bounding the cost of our solution.
  Formally, let \(s(k)\) be the number of steps performed when inserting the \(k\)-th value \(x_k\) of the input.
  If \(A[h(x_k)]\) is free, then \(s(k) = 1\); otherwise \(s(k) = i+1\), when \(x_k\) is inserted into cell \(A[h(x_k) + i \mod \gamma n]\).

  Now consider the linear probing process \(\AA_{LP}\) inserting \(n\) elements into a table \(T\) of size \(\beta n\) (that is, without the buffer space).
  Analogously to \(s(i)\) we let \(s_{LP}(i)\) be the number of steps performed by \(\AA_{LP}\) when inserting the \(i\)-th element.
  We couple the processes \(\AA\) and \(\AA_{LP}\) so that they encounter the same stream of indices \(h(x_1), h(x_2), \ldots, h(x_n)\) during execution.
  We then have \(s(i) \leq s_{LP}(i)\).
  Indeed, the only difference between the processes is that \(\AA\) has an extra \(\alpha n\) cells of buffer space to prevent wraparound at the end of \(A\).
  More generally, we thus have \(\Ep{s(i)} \leq \Ep{s_{LP}(i)}\). Combining the pieces we get that:

\begin{restatable}{claim}{stochasticLargeCost}\label{stochasticLargeCost}
$\Ep{\AA(X)} \leq O\left(1 + \frac{1}{\beta n} \sum_{i=1}^n\Ep{s(i)}\right).$%
\end{restatable}

The proof can be found in \Cref{sec:stochastic-large-appendix}. We then invoke the following classic result by Knuth \cite{knuth63linprobe} to bound \(\sum_i \Ep{s_{LP}(i)}\):
\begin{theorem}[\cite{knuth63linprobe}]\label{knuthLinprobe}
  Consider the process of inserting \((1-\eps)m\) elements into an array of size \(m\) by linear probing.
  When hash values are assigned uniformly and independently,
  \[\sum_{i=1}^{(1-\eps)m} \Ep{s_{LP}(i)} \in O\left((1-\eps) m \cdot \frac{1}{\eps}\right) \, .\]
\end{theorem}

To apply \cref{knuthLinprobe}, set \(m = \beta n\) and \(\eps = (\beta - 1)/\beta\) such that \((1-\eps) m = n\), and thus
$\sum_{i=1}^n\Ep{s(i)} \leq \sum_{i=1}^n \Ep{s_{LP}(i)} \in O\left(n \cdot \left(\frac{\beta}{\beta - 1}\right)\right)$.
As \(\beta-1 = \frac{9}{10} \cdot (\gamma - 1)\), we have
$\Ep{\AA(X)} \in O\left(1 + \frac{1}{\beta} \cdot \frac{\beta}{\beta-1}\right) 
\subseteq O\left(1 + \frac{1}{\gamma-1} \right)$, 
proving \cref{stochasticLarge}.\qed

%% file: stochastic-large-appendix.tex
\subsection{Proofs from \S\,\ref{sec41}}
\label{sec:stochastic-large-appendix}

In this subsection we prove the following claim characterizing the cost of the algorithm presented in \S\,\ref{sec41}.
\stochasticLargeCost*

A central concept for our analysis is that of the \emph{run}.
A run denotes a maximal contiguous subarray of filled cells \(A[i]\) through \(A[j]\) such that an element \(x \in \left[\frac{i}{\beta n},\, \frac{j+2}{\beta n}\right)\) will be placed in cell \(A[j+1]\).
Generally, long runs are to be avoided as they make it more likely that an element \(x\) is placed far away from \(A[h(x)]\), which makes it more difficult to bound the cost incurred by neighboring elements.

\mypara{The issue of wraparound.}
When analyzing linear probing one can often disregard the issue of wraparound as a technical corner case which has no impact on the procedure at hand other than to complicate notation and definitions.
For our use case however, wraparound would mean that we place a large value at the start of the array -- a region that we, intuitively, have otherwise been filling with small values.

Wraparound could thus incur quite a penalty, and we will prove that, w.h.p., no wraparound will happen. In this way, wraparounds can be disregarded in the following arguments.
Formally, let \(\calW\) denote the event that \(\AA\) experiences wraparound during execution, so that some element \(x\) is placed in a cell of index smaller than \(h(x)\).
We then compute the expected cost of \(\AA(X)\) as
\(\Ep{\AA(X)} = \Ep{\AA(X) \cdot \indicator{\calW}} + \Ep{\AA(X) \cdot \indicator{\neg \calW}}\).
\begin{claim}\label{stochLargeWrap}
\[\Ep{\AA(X) \cdot \indicator{\calW}} \leq O\left(\frac{1}{n}\right) \, .\]
\end{claim}
\begin{proof}
    First, note that the cost of \emph{any} solution will be at most \(n\). Hence the claim follows when we have shown that \(\Prp{\calW} \leq O\left(\frac{1}{n^2}\right)\).

    As each run starts in one of the first \(\beta n\) cells, some run needs to be of length at least \(\alpha n\) to cover the buffer at the end of \(A\), which is a necessary condition for wraparound.
    For such a long run to occur, there must in turn exist an interval \(I \subset [0, \beta n]\) of length \(\alpha n\) such that there are at least \(\alpha n\) values \(x \in X\) with \(h(x) \in I\). Denote the number of such elements by \(X_I\).

    As \(X_I\) is the sum of independent trials, each with success probability \(\alpha/\beta\), we can bound the probability of \(X_I\) exceeding \(\alpha n\) by a Chernoff bound.
    Setting \(\delta = \beta - 1\) and \(\mu = \Ep{X_I} = n \cdot \alpha/\beta\) we have, for any fixed \(\alpha, \beta\),
    \[\Prp{X_I > \alpha n} = \Prp{X_I > (1+\delta)\mu} \leq \exp\left(\frac{- \delta^2 \mu}{2 + \delta} \right) = \exp\left(\frac{- (\beta-1)^2 \frac{\alpha}{\beta} n}{1 + \beta} \right) \leq O\left(\frac{1}{n^4}\right) \, .\]
    By a union bound over all \(\beta n\) intervals we get \(\Prp{\calW} \leq O(\beta/n^3) \leq O(1/n^2)\).
\end{proof}

\mypara{Bounding the cost of the solution.}
We account for three types of cost in \(\AA(X)\): If the \(i\)-th element is inserted in a cell right before the start of a run, the \(i\)-th step is said to incur a \emph{merge} cost, denoted \(mer(i)\), as it potentially merges two runs.
If the \(i\)-th element is inserted in the cell following a run, the \(i\)-th step is said to incur an \emph{extend} cost, denoted \(ext(i)\), as it extends the run.
If the \(i\)-th insertion does not give rise to a merge- (resp.\ extend-) cost, we set \(mer(i) = 0\) (resp.\ \(ext(i) = 0\)).
Finally we have to account for the cost between values \(A[i]\) and \(A[j]\) separated by one or more empty cells.
This cost is not charged to a specific insertion performed by the algorithm but is instead viewed as a property of the solution, \(sep(\AA(X))\).

We bound the expected size of each type of cost in the following three claims, and together they account for the value of the solution produced by \(\AA\).
\begin{claim}\label{stochLargeMer}
  \begin{align*}
    \sum_{i=1}^{n} \Ep{mer(i) \cdot \indicator{\neg \calW}} \leq \frac{1}{\beta n}\sum_{i=1}^n \Ep{s(i)} + \frac{1}{\beta} \, .
  \end{align*}
\end{claim}
\begin{proof}
  Let \(x_i\) be the \(i\)-th element of \(X\) and let \(k\) be the index where \(x_i\) is placed.
  Then \(s(i) = 1 + k - h(x_i)\), due to our assumption that no wraparound occurs.
  Let \(y = A[k+1]\) be the value in the neighbouring cell, which was inserted before \(x_i\).
  Then \(mer(i) = \size{y - x_i}\).

  As cell \(A[k]\) was empty at the time of \(y\)'s insertion we must have \(y \in \left[\frac{k+1}{\beta n} , \, \frac{k+2}{\beta n}\right)\).
  Meanwhile, \(x_i \in \left[\frac{h(x_i)}{\beta n} , \, \frac{h(x_i) +1}{\beta n} \right)\) and thus
      \[\size{y - x_i} \leq \frac{k+2 - h(x_i)}{\beta n} = \frac{s(i) + 1}{\beta n} \, .\]
      It follows that \(\sum_{i=1}^n \Ep{mer(i) \cdot \indicator{\neg \calW}} \leq \frac{1}{\beta n} \sum_{i=1}^n \Ep{s(i)} + \frac{1}{\beta}\).
\end{proof}

\begin{claim}\label{stochLargeExt}
  \[\sum_{i=1}^n \Ep{ext(i) \cdot \indicator{\neg \calW}} \leq \frac{1}{\beta n} \sum_{i=1}^n \Ep{s(i)} \, .\]
\end{claim}
\begin{proof}
Denote by \(I_k\) the \(k\)-th run in \(A\) which spans cells \(A[a_k]\) through \(A[b_k]\) before the insertion of \(x_i\), with value \(y_k = A[b_k]\) in the final position.
Assuming no wraparound occurs, we have \(a_k \leq b_k\).
In a slight abuse of notation we denote by \(h(x_i) \in I_k\) the event that \(x_i\) is appended to \(I_k\), which is equivalent to the event \(h(x_i) \in [a_k, b_k+1]\).
Assuming \(h(x_i) \in I_k\), we then have \(s(i) = 1 + (b_k+1) - h(x_i) = 2 + b_k - h(x_i)\).
Then
\begin{align*}
    ext(i) = \sum_{k=1} \size{x_i - y_k} \cdot \indicator{h(x_i) \in I_k} \leq \sum_{k=1} \frac{1}{\beta n} \left(\size{h(x_i) - h(y_k)} + 1 \right) \cdot \indicator{h(x_i) \in I_k} \, .
\end{align*}
As we are interested in \(\Ep{ext(i)}\) we wish to bound \(\Ep{\size{h(x_i) - h(y_k)} \cdot \indicator{h(x_i) \in I_k}}\) for some fixed \(k\).
We will not delve into the distribution of \(h(y_k)\) but note that the distribution of \(h(x_i)\) is independent of \(h(y_k)\).
Conditioning on \(h(y_k)\), we see that the expectation is maximized when \(h(y_k)\) is at one of the ``extremes'' \(h(y_k) = a_k\) or \(h(y_k) = b_k\).
For ease of computation we substitute \(h(y_k) = b_k + 1\) for our upper bound:
\begin{align*}
    \Ep{\size{h(x_i) - h(y_k)} \cdot \indicator{h(x_i) \in I_k} \cdot \indicator{\neg \calW}} 
    &\leq \Ep{\size{(b_k+1) - h(x_i)} \cdot \indicator{h(x_i) \in I_k}} \\
    &= \Ep{(s(i)-1) \cdot \indicator{h(x_i) \in I_k}} \, .
\end{align*}
Plugging into our previous equation we obtain
\begin{align*}
    \Ep{ext(i) \cdot \indicator{\neg \calW}} &\leq \frac{1}{\beta n} \sum_{k=1} \Ep{s(i) \cdot \indicator{h(x_i) \in I_k}} \leq \frac{1}{\beta n} \cdot \Ep{s(i)} \, .
\end{align*}
\end{proof}

\begin{claim}\label{stochLargeSep}
  \[\Ep{sep(\AA(X)) \cdot \indicator{\neg \calW}} \leq 2 \, .\]
\end{claim}
\begin{proof}
  Let the runs be numbered \(1, \dots, r\) according to their order in \(A\), and say that the \(i\)-th run covers cells \(A[a_i]\) through \(A[b_i]\) such that \(a_i \leq b_i < a_{i+1} \leq b_{i+1}\).
  With this notation \(sep(\AA(X)) = \sum_{i=1}^r \size{A[a_{i+1}] - A[b_{i}]}\).

  Fix some \(i\), let \(x = A[a_{i+1}]\) and \(y = A[b_i]\).
  First, observe that \(h(x) = a_{i+1}\).
  Had \(h(x)\) been lower, \(x\) would have been placed in a cell of smaller index (and we know that \(A[a_{i+1} - 1]\) is empty). Had \(h(x)\) been higher, \(x\) would have been placed in a cell of index at least \(a_{i+1}\) due to our assumption that no wraparound occurs.
  Hence \(A[a_{i+1}] \in \left[\frac{a_{i+1}}{\beta n} , \, \frac{a_{i+1} + 1}{\beta n} \right)\).

  Similarly, we have \(a_i \leq h(y) \leq b_i\) and thus \(A[b_i] \in \left[\frac{a_i}{\beta n}, \, \frac{b_i + 1}{\beta n}\right)\), as the element \(y\) would otherwise have been placed outside of the run.
  Thus
  \begin{align*}
    \sum_{i=1}^r \size{A[a_{i+1}] - A[b_i]} \leq \sum_{i=1}^r \frac{a_{i+1}- a_i + 1}{\beta n} = \frac{a_r - a_1 + r}{\beta n} \leq 2 \,
  \end{align*}
  as \(r \leq a_r \leq \beta n\).
\end{proof}

Summing over \cref{stochLargeWrap,stochLargeExt,stochLargeMer,stochLargeSep}, we obtain the bound stated in \Cref{stochasticLargeCost}.\qed

%% file: main.bbl
\begin{thebibliography}{10}

\bibitem{online_sort}
Anders Aamand, Mikkel Abrahamsen, Lorenzo Beretta, and Linda Kleist.
\newblock Online sorting and translational packing of convex polygons.
\newblock In Nikhil Bansal and Viswanath Nagarajan, editors, {\em Proceedings
  of the 2023 {ACM-SIAM} Symposium on Discrete Algorithms, {SODA} 2023,
  Florence, Italy, January 22-25, 2023}, pages 1806--1833. {SIAM}, 2023.
\newblock URL: \url{https://doi.org/10.1137/1.9781611977554.ch69}, \href
  {https://doi.org/10.1137/1.9781611977554.CH69}
  {\path{doi:10.1137/1.9781611977554.CH69}}.

\bibitem{arbitman2010backyard}
Yuriy Arbitman, Moni Naor, and Gil Segev.
\newblock Backyard cuckoo hashing: Constant worst-case operations with a
  succinct representation.
\newblock In {\em 2010 IEEE 51st Annual symposium on foundations of computer
  science}, pages 787--796. IEEE, 2010.

\bibitem{ausiello2001algorithms}
Giorgio Ausiello, Esteban Feuerstein, Stefano Leonardi, Leen Stougie, and
  Maurizio Talamo.
\newblock Algorithms for the on-line travelling salesman 1.
\newblock {\em Algorithmica}, 29:560--581, 2001.

\bibitem{balkanski2023power}
Eric Balkanski, Yuri Faenza, and No{\'e}mie P{\'e}rivier.
\newblock The power of greedy for online minimum cost matching on the line.
\newblock In {\em Proceedings of the 24th ACM Conference on Economics and
  Computation}, pages 185--205, 2023.

\bibitem{Beardwood_Halton_Hammersley_1959}
Jillian Beardwood, J.~H. Halton, and J.~M. Hammersley.
\newblock The shortest path through many points.
\newblock {\em Mathematical Proceedings of the Cambridge Philosophical
  Society}, 55(4):299–327, 1959.
\newblock \href {https://doi.org/10.1017/S0305004100034095}
  {\path{doi:10.1017/S0305004100034095}}.

\bibitem{bender2002two}
Michael~A Bender, Richard Cole, Erik~D Demaine, Martin Farach-Colton, and Jack
  Zito.
\newblock Two simplified algorithms for maintaining order in a list.
\newblock In {\em European Symposium on Algorithms}, pages 152--164. Springer,
  2002.

\bibitem{bender2022online}
Michael~A Bender, Alex Conway, Mart{\'\i}n Farach-Colton, Hanna Koml{\'o}s,
  William Kuszmaul, and Nicole Wein.
\newblock Online list labeling: Breaking the log 2 n barrier.
\newblock In {\em 2022 IEEE 63rd Annual Symposium on Foundations of Computer
  Science (FOCS)}, pages 980--990. IEEE, 2022.

\bibitem{DBLP:conf/stoc/BenderFKKL22}
Michael~A. Bender, Martin Farach{-}Colton, John Kuszmaul, William Kuszmaul, and
  Mingmou Liu.
\newblock On the optimal time/space tradeoff for hash tables.
\newblock In Stefano Leonardi and Anupam Gupta, editors, {\em {STOC} '22: 54th
  Annual {ACM} {SIGACT} Symposium on Theory of Computing, Rome, Italy, June 20
  - 24, 2022}, pages 1284--1297. {ACM}, 2022.
\newblock \href {https://doi.org/10.1145/3519935.3519969}
  {\path{doi:10.1145/3519935.3519969}}.

\bibitem{bjelde2020tight}
Antje Bjelde, Jan Hackfeld, Yann Disser, Christoph Hansknecht, Maarten Lipmann,
  Julie Mei{\ss}ner, Miriam Schl{\"o}ter, Kevin Schewior, and Leen Stougie.
\newblock Tight bounds for online tsp on the line.
\newblock {\em ACM Transactions on Algorithms (TALG)}, 17(1):1--58, 2020.

\bibitem{Borodin}
Allan Borodin and Ran El{-}Yaniv.
\newblock {\em Online computation and competitive analysis}.
\newblock Cambridge University Press, 1998.

\bibitem{dietz1982maintaining}
Paul~F Dietz.
\newblock Maintaining order in a linked list.
\newblock In {\em Proceedings of the fourteenth annual ACM symposium on Theory
  of computing}, pages 122--127, 1982.

\bibitem{Few_1955}
L.~Few.
\newblock The shortest path and the shortest road through n points.
\newblock {\em Mathematika}, 2(2):141--144, 1955.
\newblock \href {https://doi.org/10.1112/S0025579300000784}
  {\path{doi:10.1112/S0025579300000784}}.

\bibitem{FiatW99}
Amos Fiat and Gerhard~J. Woeginger.
\newblock On-line scheduling on a single machine: Minimizing the total
  completion time.
\newblock {\em Acta Informatica}, 36(4):287--293, 1999.
\newblock URL: \url{https://doi.org/10.1007/s002360050162}, \href
  {https://doi.org/10.1007/S002360050162} {\path{doi:10.1007/S002360050162}}.

\bibitem{gupta2019stochastic}
Anupam Gupta, Guru Guruganesh, Binghui Peng, and David Wajc.
\newblock Stochastic online metric matching.
\newblock In {\em 46th International Colloquium on Automata, Languages, and
  Programming (ICALP 2019)}. Schloss-Dagstuhl-Leibniz Zentrum f{\"u}r
  Informatik, 2019.

\bibitem{gupta2012online}
Anupam Gupta and Kevin Lewi.
\newblock The online metric matching problem for doubling metrics.
\newblock In {\em Automata, Languages, and Programming: 39th International
  Colloquium, ICALP 2012, Warwick, UK, July 9-13, 2012, Proceedings, Part I
  39}, pages 424--435. Springer, 2012.

\bibitem{itai1981sparse}
Alon Itai, Alan~G Konheim, and Michael Rodeh.
\newblock A sparse table implementation of priority queues.
\newblock In {\em International Colloquium on Automata, Languages, and
  Programming}, pages 417--431. Springer, 1981.

\bibitem{DBLP:journals/siamdm/Karloff89}
Howard~J. Karloff.
\newblock How long can a euclidean traveling salesman tour be?
\newblock {\em {SIAM} J. Discret. Math.}, 2(1):91--99, 1989.
\newblock \href {https://doi.org/10.1137/0402010} {\path{doi:10.1137/0402010}}.

\bibitem{knuth63linprobe}
Donald~E. Knuth.
\newblock Notes on open addressing.
\newblock Unpublished memorandum. See
  \url{http://citeseer.ist.psu.edu/knuth63notes.html}, 1963.

\bibitem{Knuth3}
Donald~E. Knuth.
\newblock {\em The Art of Computer Programming, Volume 3: (2nd Ed.) Sorting and
  Searching}.
\newblock Addison Wesley Longman Publishing Co., Inc., Redwood City, CA, USA,
  1998.

\bibitem{megow2012power}
Nicole Megow, Martin Skutella, Jos{\'e} Verschae, and Andreas Wiese.
\newblock The power of recourse for online mst and tsp.
\newblock In {\em Automata, Languages, and Programming: 39th International
  Colloquium, ICALP 2012, Warwick, UK, July 9-13, 2012, Proceedings, Part I
  39}, pages 689--700. Springer, 2012.

\bibitem{motwani1995}
Rajeev Motwani and Prabhakar Raghavan.
\newblock Randomized algorithms.
\newblock In Mikhail~J. Atallah, editor, {\em Algorithms and Theory of
  Computation Handbook}, Chapman {\&} Hall/CRC Applied Algorithms and Data
  Structures series. {CRC} Press, 1999.
\newblock URL: \url{https://doi.org/10.1201/9781420049503-c16}, \href
  {https://doi.org/10.1201/9781420049503-C16}
  {\path{doi:10.1201/9781420049503-C16}}.

\bibitem{tsp_hard}
Christos~H Papadimitriou.
\newblock The euclidean travelling salesman problem is np-complete.
\newblock {\em Theoretical computer science}, 4(3):237--244, 1977.

\bibitem{pruhs2004online}
Kirk Pruhs, Jir{\'{\i}} Sgall, and Eric Torng.
\newblock Online scheduling.
\newblock In Joseph~Y.{-}T. Leung, editor, {\em Handbook of Scheduling -
  Algorithms, Models, and Performance Analysis}. Chapman and Hall/CRC, 2004.
\newblock URL:
  \url{http://www.crcnetbase.com/doi/abs/10.1201/9780203489802.ch15}, \href
  {https://doi.org/10.1201/9780203489802.CH15}
  {\path{doi:10.1201/9780203489802.CH15}}.

\bibitem{raghvendra2018optimal}
Sharath Raghvendra.
\newblock Optimal analysis of an online algorithm for the bipartite matching
  problem on a line.
\newblock In {\em 34th International Symposium on Computational Geometry (SoCG
  2018)}. Schloss Dagstuhl-Leibniz-Zentrum fuer Informatik, 2018.

\bibitem{bwac_book}
Tim Roughgarden, editor.
\newblock {\em Beyond the Worst-Case Analysis of Algorithms}.
\newblock Cambridge University Press, 2020.
\newblock \href {https://doi.org/10.1017/9781108637435}
  {\path{doi:10.1017/9781108637435}}.

\bibitem{saks2018online}
Michael Saks.
\newblock Online labeling: Algorithms, lower bounds and open questions.
\newblock In {\em International Computer Science Symposium in Russia}, pages
  23--28. Springer, 2018.

\end{thebibliography}
